\documentclass[copyright,sharealike,creativecommons]{eptcs}
\providecommand{\event}{CL\&C 2010}
\title{Dialectica Interpretation with Marked Counterexamples}
\author{Trifon Trifonov\thanks{The author gratefully acknowledges financial
support by the Bulgarian National Science Fund within project DO 02-102/23.04.2009 and by the European Social Fund within project BG 051PO001-3.3.04/28.08.2009.}
\institute{Faculty of Mathematics and Informatics, Sofia University}
	\email{triffon@fmi.uni-sofia.bg}}

\def\titlerunning{Dialectica Interpretation with Marked Counterexamples}
\def\authorrunning{Trifon Trifonov}

\usepackage[utf8]{inputenc}
\usepackage{amsmath}
\usepackage{amssymb}
\usepackage{amsthm}
\usepackage{makecmds}

\newtheorem{theorem}{Theorem}

\newtheorem{lemma}{Lemma}

\theoremstyle{definition}
\newtheorem{definition}{Definition}

\theoremstyle{remark}

\newcommand{\ve}{\varepsilon}

\newcommand{\pair}[1]{\left\langle #1\right\rangle}

\newcommand{\abs}[1]{\left\vert #1 \right\vert}

\newcommand{\set}[1]{\left\{ #1\right\}}

\newcommand{\pstar}{$(*)$ }

\newcommand{\enumr}{\renewcommand{\theenumi}{\roman{enumi}}}

\newcommand{\bnor}{\; |\;}
\newcommand{\bndef}{\quad \mathtt{::=} \quad}

\newcommand{\imp}{\rightarrow}
\newcommand{\eqv}{\leftrightarrow}

\newcommand{\leqv}{\eqv}

\newcommand{\true}{\mathsf{tt}}
\newcommand{\false}{\mathsf{ff}}

\newcommand{\yields}{\vdash}

\newcommand{\subst}[2]{\left[#1 := #2\right]}

\newcommand{\quantifier}[2]{}

\newcommand{\forallsymbol}{\forall}
\newcommand{\lambdasymbol}{\lambda}
\newcommand{\existssymbol}{\exists}
\newcommand{\existsclsymbol}{\tilde\exists}

\newcommand{\all}{\quantifier\forallsymbol}
\newcommand{\lam}{\quantifier\lambdasymbol}
\newcommand{\ex}{\quantifier\existssymbol}
\newcommand{\excl}{\quantifier\existsclsymbol}

\newcommand{\nosubscriptquantifiers}{
  \renewcommand{\quantifier}[2]{##1 {##2}\,}
}

\nosubscriptquantifiers

\newcommand{\FV}{\mathsf{FV}}
\newcommand{\BV}{\mathsf{BV}}

\newcommand{\PP}{\mathcal{P}}

\newcommand{\HA}{\mathrm{HA}}
\newcommand{\NA}{\mathsf{NA}}

\newcommand{\Cases}{\mathsf{Cases}}
\newcommand{\efq}{\mathsf{efq}}

\newcommand{\Truth}{\mathsf{AxT}}
\newcommand{\AxTrue}{\Truth}

\newcommand{\Ind}{\mathsf{Ind}}

\newcommand{\atom}[1]{\mathrm{at}(#1)}
\newcommand{\False}{\mathrm{F}}

\newcommand{\FA}{\mathsf{FA}}

\newcommand{\To}{\Rightarrow}

\newcommand{\Times}{\times}

\newcommand{\nat}{\mathtt{nat}}
\newcommand{\bool}{\mathtt{bool}}

\newcommand{\Succ}{\mathsf{S}}
\newcommand{\pl}{\llcorner}
\newcommand{\pr}{\lrcorner}

\newcommand{\inhab}{\sqcup}

\newcommand{\R}{\mathcal{R}}
\newcommand{\Rec}{\R}

\newcommand{\Cas}{\mathcal{C}}

\newcommand{\lett}[2]{\textbf{let }#1\textbf{ in }#2}

\newcommand{\nulltype}{\ve}
\newcommand{\nullterm}{\ve}

\newcommand{\context}[2]{#1[#2]}
\newcommand{\hole}{\context{}{}}
\newcommand{\holetype}{\diamond}

\newcommand{\extapply}[2]{#1 \circ #2}

\usepackage{namespc}

\newcommand{\dplus}[1]{{#1}^+}
\newcommand{\dminus}[1]{{#1}^-}

\newcommand{\dialast}[1]{{#1}^\ast}

\newcommand{\typesymbol}{\tau}
\newcommand{\typeextr}[2]{#1{\typesymbol}(#2)}

\newcommand{\typep}{\typeextr \dplus}
\newcommand{\typen}{\typeextr \dminus}

\newcommand{\typeast}{\typeextr \dialast}
\newcommand{\typea}{\typeast}

\newcommand{\intd}[3]{\abs{#1}^{#2}_{#3}}

\newcommand{\et}[1]{[\![#1]\!]}
\newcommand{\etmark}[3]{#2{#1{#3}}}

\newcommand{\etdp}{\etmark\et\dplus}
\newcommand{\etdn}{\etmark\et\dminus}

\newcommand{\contract}[3]{#2\,{\stackrel{#1}{\bowtie}}\,#3}
\newcommand{\contractterm}[1]{T^{#1}_{\contract{}{}{}}}

\namespace{linextr}{%
\def\typesymbol{\sigma}%
\newcommand{\etds}[1]{\{\!|#1|\!\}}%
\newcommand{\etdsp}{\etmark\etds\dplus}%
\newcommand{\etdsn}{\etmark\etds\dminus}%
}{}

\namespace{markcounter}{%
\def\typesymbol{\rho}%
\newcommand{\marktype}{\bool^\bot}%
\newcommand{\markby}[2]{#1\blacktriangleright #2}%
\newcommand{\marked}[1]{{#1}^\multimap}
\newcommand{\typem}{\typeextr\marked}
}{}%
\usepackage{multirow}
\usepackage{url}

\newcommand{\IndBool}{\Cases}
\newcommand{\IndNat}{\Ind}
\newcommand{\RecNat}{\Rec}

\renewcommand{\bool}{\mathtt{B}}
\renewcommand{\nat}{\mathtt{N}}

\allowdisplaybreaks

\newcommand{\tprod}{\Times}
\newcommand{\ite}[3]{\Cas\,#1\,#2\,#3}

\newcommand{\size}[1]{\lceil #1\rceil}
\newcommand{\msl}[1]{\lceil\!\lceil #1 \rceil\!\rceil}
\newcommand{\reduces}{\stackrel r{\mapsto}}
\newcommand{\redeq}{\stackrel r=}

\begin{document}
  \maketitle
  
\begin{abstract}
  G\"odel's functional ``Dialectica'' interpretation can be used to extract functional programs from non-constructive proofs in arithmetic by employing two sorts of higher-order witnessing terms: positive realisers and negative counterexamples. In the original interpretation decidability of atoms is required to compute the correct counterexample from a set of candidates. When combined with recursion, this choice needs to be made for every step in the extracted program, however, in some special cases the decision on negative witnesses can be calculated only once. We present a variant of the interpretation in which the time complexity of extracted programs can be improved by marking the chosen witness and thus avoiding recomputation. The achieved effect is similar to using an abortive control operator to interpret computational content of non-constructive principles.
\end{abstract}

\section{Introduction}
G\"odel's Dialectica interpretation \cite{goedel} is one of the first systematic methods for obtaining computational content from proofs in classical arithmetic. Different variants of the interpretation have been proposed to aid ``proof mining'' --- looking for constructive information inside what appears to be a non-constructive argument (e.g. \cite{kohlenbook}). Particularly interesting is the case where terms of non-ground types (i.e., functional programs) are being automatically obtained from a weak existence proof. The extracted algorithms usually calculate the witness for the existential quantifier in an indirect manner, often obscure and surprising. There are also a number of competing techniques for program extraction like (refined) $A$-translation \cite{friedman,bbs}, Krivine's realisability, control operators \cite{griffin}, the $\lambda\mu$-calculus \cite{parigot} and others. All of these methods systematically find correct programs, but the relations between them are still being investigated.

Another topic of ongoing research is whether such approaches are feasible for practical extraction of sufficiently efficient correct algorithms. Even though many automatic software systems for handling large proof objects are being actively developed, additional work is needed to identify and remove possible redundancies, so that the extracted programs are more readable, shorter and faster. Examples of such optimisations include uniform decorations \cite{uniformha,danhernest,decorating}, soundness-preserving program transformations \cite{makaroventcs}, avoiding syntactic repetition \cite{quasilin}. In the present paper we suggest another such technique for the Dialectica interpretation, which marks computed counterexamples that are determined to be valid. We demonstrate how in certain cases this approach can reduce the average time complexity of the obtained program by terminating recursive search immediately after a correct counterexample is found. The suggested change is an extension of the interpretation variant given in \cite{quasilin}. The reason is that even though counterexample marking can be formally applied directly to the original Dialectica interpretation, its practical effects are only visible when (at least) syntactic repetition is avoided.

\section{Negative Arithmetic}

We work in a restriction of Heyting Arithmetic with finite types (denoted $\HA^\omega$ in \cite{troelstra}) to the language of $\imp$ and $\forall$. We refer to the resulting system as Negative Arithmetic ($\NA^\omega$).

\begin{definition}
 Types $(\rho,\sigma)$, (object) terms $(s,t)$ and formulas $(A,B)$ are defined as follows:
\begin{eqnarray*}
 \rho,\sigma &\bndef&  \bool \bnor \nat \bnor \alpha \bnor
\rho \To \sigma \bnor \rho \Times \sigma\\
 s,t & \bndef &x^\rho \bnor (\lam {x^\rho} t^\sigma)^{\rho\To\sigma} \bnor (s^{\rho\To\sigma}t^\rho)^\sigma 
\bnor \pair{s^\rho,t^\sigma}^{\rho\Times\sigma} \bnor (t^{\rho\Times\sigma}\pl)^\rho \bnor (t^{\rho\Times\sigma}\pr)^\sigma \bnor \\
&& \true^\bool \bnor \false^\bool \bnor 0^\nat \bnor \Succ^{\nat\To\nat} \bnor
\Cas^{\bool\To\sigma\To\sigma\To\sigma} \bnor \RecNat^{\nat\To\sigma\To(\nat\To\sigma\To\sigma)\To\sigma} \\
A,B &\bndef & \atom{t^\bool} \bnor A \imp B \bnor \all {x^\rho} A
\end{eqnarray*}
The base types of booleans $\bool$ and natural numbers $\nat$ are equipped with the usual constructors and structural recursor constants. Here $x$ denotes a typed object variable and $\alpha$ denotes a type variable. Freely occurring type variables allow for a restricted form of polymorphism.
The sets of free variables $\FV(t)$, $\FV(A)$ and bound variables $\BV(t)$, $\BV(A)$
are defined inductively as usual. Substitution of terms for object variables $s\subst{x}{t}$, $A\subst{x}{t}$
 is by default assumed to be capture-free with respect to abstraction and quantification.
\end{definition}

The operational semantics of object terms are given by the usual $\beta$-reduction rules and computation rules for the recursor constants:
\begin{equation*}
\begin{array}{ll}
\begin{array}{lll}
\pair{s,t}\pl &\reduces& s\\
\pair{s,t}\pr &\reduces& t\\
(\lam x s)t &\reduces& s\subst{x}{t}
\end{array}
&\quad
\begin{array}{llllll}
\ite\true{t_1}{t_2} &\reduces& t_1 &\quad \RecNat\,0\,s\,t &\reduces& s\\
\ite\false{t_1}{t_2} &\reduces& t_2 &\quad \RecNat\,(\Succ n)\,s\,t &\reduces& t\,n\,(\RecNat\,n\,s\,t)
\end{array}
\end{array}
\end{equation*}
We will make use of the following ``let'' notation for a $\beta$-redex:
\begin{gather*}
 \lett{x:=t}s \qquad := \qquad (\lam x s)t.
\end{gather*}

We express derivations in a natural deduction system with a similar syntax to that of object terms to stress the Curry-Howard correspondence. Proof terms are typed by their conclusion formulas and are built from \emph{assumption variables}.

\begin{definition}
  \emph{Proof terms} $(M,N)$ of $\NA^\omega$ are defined as follows:
\begin{eqnarray*}
 M,N & \bndef &u^A \bnor (\lam {u^A} M^B)^{A\imp B} \bnor (M^{A\imp B}N^A)^B \bnor \\
&\text{\pstar} & (\lam {x^\rho} M^{A(x)})^{\all {x^\rho}A(x)} \; \bnor
(M^{\all {x^\rho} A(x)}t)^{A(t)} \bnor \\
&& \AxTrue: \atom{\true} \bnor 
\IndBool^{A(b)}:\all {b^\bool}\big(A(\true) \imp A(\false) \imp A(b)\big) \bnor\\
&&\IndNat^{A(n)}: \all {n^\nat}\big(A(0) \imp \all {n^\nat}(A(n) \imp A(\Succ n)) \imp A(n)\big)
\end{eqnarray*}
with the usual variable condition \pstar that the object variable $x$ does not occur freely in any of the open assumptions of $M$. The sets of free variables $\FV(M)$ and free (open) assumption variables $\FA(M)$ as well as capture-free substitutions $M\subst{x}{t}$ and $M\subst{u}{N}$ are defined inductively as usual.
\end{definition}

The \emph{truth axiom} $\AxTrue$ defines the logical meaning of $\atom{\cdot}$ and allows us to consider any boolean valued function defined in our term system as a decidable predicate. When we write for example $n=m$, we actually mean $\atom{\mathrm{Eq}\,n\,m}$, where $\mathrm{Eq}^{\nat\To\nat\To\bool}$ is a term defining the decidable equality for natural numbers.

In our negative language defining falsity as $\False := \atom{\false}$ already gives us the full power of classical logic. In particular, for a formula $A$ we can prove \emph{Ex falso quodlibet} ($\efq$): $\yields \False \imp A$ and \emph{Stability}: ${\yields ((A\imp \False)\imp \False) \imp A}$ by meta induction on $A$, using $\AxTrue$ and $\IndBool$ for the base case.
We will thus use the abbreviations $\neg A := A \imp \False$ and $\excl {x^\rho} A := \neg\all {x^\rho} \neg A$.

The term system in consideration is essentially G\"odel's T and the reduction relation $\reduces$ is well-known to be strongly normalising and confluent. Thus, instead of insisting that object terms appearing in formulas of proof rules match exactly, we require them only to have the same $\eta$-long normal form; this equality will be denoted by $\redeq$. Note that we make no such assumption for extracted programs or for proof terms themselves.

{\bf Notation.}
For technical convenience we will use $\nulltype$ for denoting a special \emph{nulltype}, i.e., lack of computational content. By abuse of notation we also use $\nullterm$ to denote all terms of nulltype. We stipulate that the following simplifications are always carried out implicitly:
\begin{equation}\
\newcommand{\simpl}{\rightsquigarrow}
\begin{array}{c@{\;}c}
\begin{array}{r@{\,}l@{\quad}r@{\,}l@{\quad}r@{\,}l}
 \rho\tprod\ve &\simpl \rho, & t^{\rho\tprod\ve}\pl &\simpl t,&\pair{t,\ve} &\simpl t,\\
 \ve\tprod\rho &\simpl \rho, & t^{\ve\tprod\rho}\pr &\simpl t,&\pair{\ve,t} &\simpl t,\\
&&&&&
\end{array}
&
\begin{array}{r@{\,}l@{\quad}r@{\,}l@{\quad}r@{\,}l}
 \rho\To\ve &\simpl \ve,&            \lam x\ve &\simpl \ve, & \ve t &\simpl \ve\\
 \ve\To\rho &\simpl \rho, &           \lam {x^\ve}t &\simpl t,& t\ve &\simpl t\\
 & &           \all {x^\ve} A &\simpl A, &M\ve &\simpl M
\end{array}
\end{array}
\tag{$\ve$} \label{eq:epsilon}
\end{equation}
Consequently, to simplify presentation all $\nullterm$ terms will be silently omitted, as they hold no computational content.

\section{Quasi-linear Dialectica interpretation}

\begin{linextr}
We will shortly outline the variant of the Dialectica interpretation, which was presented in \cite{quasilin}. It allows extraction of more efficient programs by avoiding syntactic repetition of subterms. In particular, it turns out that the size of the extracted terms depends almost linearly on the proof size. The present paper will build upon this interpretation to improve efficiency of recursion even further.

The general idea behind avoidance of syntactic repetition is to factor out common subterms as much as possible in the positive and negative content during the extraction process. To achieve this we use \emph{definition contexts} --- a tool, which allows to gradually accumulate the common part of all witnesses of a given proof. In order to apply this technique, the definition of the Dialectica computational types needs to be slightly revised so that we use uncurried function types instead of curried ones, because both the partial and the full application of an uncurried function to a variable increase the term size with a constant\footnote{In contrast, full application of a curried function needs a variable for each parameter.}.

We start with some preliminary notations. We use $\size\cdot$ to denote size of terms, formulas and proofs. For a proof $M$ we define its \emph{maximal sequent length} $\msl M$ as $\max_{N\leq M} \abs{\FA(N)}$,  where $N\leq M$ is the subproof relation. The rest of the needed definitions are presented below.

\begin{definition}
  Let us fix a type variable $\holetype$ and an object variable $\hole$ of type $\holetype$, which will be referred to as ``a hole''. A \emph{definition context} $E$ is a term built by the following rules:
\begin{align*}
  E\bndef \hole^\holetype \bnor (E^{\rho\To\sigma}t^\rho)^\sigma \bnor (\lam {x^\rho} E^\sigma)^{\rho\To\sigma},
\end{align*}
where $t$ does not contain the type $\holetype$. For a definition context $E^\rho$ and term $t^\sigma$, we define the term $\context Et$ ($t$ in the context $E$) as $E\subst\holetype\sigma\subst\hole t$, where, contrary to our usual convention, the free variables of $t$ are allowed to be bound by abstractions in $E$.
\end{definition}

\begin{definition}
We define the partial application of the (uncurried) function $f$ to the term $t$ as
\begin{align*}
 \extapply{f^{\rho\To\tau}}{t^\rho} &:= ft\\
 \extapply{f^{\rho\Times\sigma\To\tau}}{t^\rho} &:= \lam {x^\sigma} f\pair{t,x}, \text{ where $x$ is a fresh variable.}
\end{align*}
\end{definition}

\begin{definition}
We extend the projection operations $\pl$ and $\pr$ to functions:
\begin{gather*}
 f^{\rho\To\sigma\Times\tau}\pl := \lam {x^\rho} fx\pl,\qquad f^{\rho\To\sigma\Times\tau}\pr := \lam {x^\rho} fx\pr.
\end{gather*}
\end{definition}

\begin{definition}
 For a formula $A$ we define the positive and negative computational types ($\typep A$ and $\typen A$). We will also denote $\typea A := \typen A \To \typep A$.
We define:
\begin{align*}
 &\typep {\atom b} := \nulltype,&&\typen{\atom b} := \nulltype,\\
 &\typep{A\imp B} := \typep B \Times \typen A,&&\typen{A\imp B} := \typea A \Times \typen B,\\
 &\typep{\all {x^\rho} A} := \typep A,&&\typen{\all {x^\rho} B} := \rho \Times \typen B.
\end{align*}
\end{definition}

\begin{definition}
 For $r:\typea A$, $s:\typen A$ we define $\intd Ars$ as follows:
 \begin{align*}
  \intd{\atom b}{}{} &:= \atom b, \qquad     \intd{\all x A}rs := \intd {A(s\pl)}{\extapply r{s\pl}}{s\pr},\\
  \intd{A\imp B}rs &:= \intd A{s\pl}{rs\pr} \imp \intd B{(\extapply r{s\pl})\pl}{s\pr}.\\
 \end{align*}
\end{definition}

The soundness theorem for the new variant of the interpretation follows a similar pattern to the usual soundness proof. On every inductive step we define:
\begin{enumerate}
 \item a definition context $\et M:\typen A\To\holetype$
 \item a context-dependent positive witnessing term $\etdp M:\typep A$
 \item context-dependent negative witnessing terms $\etdn M_i:\typen {C_i}$
\end{enumerate}
The final extracted term will be obtained by putting the context-dependent terms inside the context:
\begin{gather*}
 \etds M := \context{\et M}{\pair{\etdp M, \ldots, \etdn M_i,\ldots}}.
\end{gather*}
We will refer to the separate components put in the context as follows:
\begin{align*}
 \etdsp M := \context{\et M}{\etdp M}, \qquad \etdsn M_i := \context{\et M}{\etdn M_i}.
\end{align*}

\begin{theorem}[Soundness of quasi-linear extraction]\label{thm:quasilin}
   Let $\PP:A$ be a proof in $\NA^\omega$ from assumptions $u_i:C_i$. Let $x_i:\typea{C_i}$ and $y_A:\typen A$ be fresh variables. Then there is a term $\etds \PP$, satisfying the following conditions:
\begin{enumerate}
\enumr
 \item we can prove $\intd A {\etdsp \PP}{y_A}$ from $\intd {C_i} {x_i} {\etdsn \PP_i{y_A}}$,
 \item $\FV(\etds \PP) \subseteq \FV(\PP) \cup \set{x_i}$,
 \item $\size{\etds \PP} \leq K(\size \PP + {\msl \PP}^2)$ for a fixed constant $K$, not depending on $\PP$.
\end{enumerate}
\end{theorem}

\section{A special case of recursion}
\label{subsec:countercheck}

A specific feature of the Dialectica interpretation which allows to embed classical logic into a quantifier-free constructive system is the extraction of counterexamples. In our negative language, in order to prove $\excl x A$ we need to use the assumption $\all x \neg A$ to derive a contradiction. The non-trivial use of classical logic comes where we use this assumption more than once. In the extracted term this corresponds to deciding between counterexamples by checking the validity of the quantifier-free translation $\intd A x y$. An extreme example of this phenomenon is the interpretation of induction\footnote{Here we refer to the full induction rule, not to the commonly considered assumptionless induction rule.}, which corresponds to using the induction hypothesis an unbounded number of times. This is reflected by a case distinction on every recursive step in the recursively defined programs for computing counterexamples for open assumptions. However, there is a special case of the induction scheme in which a case distinction on every step is redundant and, moreover, can lead to an unnecessary increase of complexity.

Let $\PP := \Ind_{\nat,A(n)}\,n\,M^{A(0)}\,(\lam {n,v^{A(n)}} N^{A(n+1)})$ be a proof by induction from assumptions $u_i:C_i$. Consider the case where $\typen A = \ve$. For the sake of simplicity let us assume that we have only one open assumption $u:C$ and let us omit all indices. By the usual soundness theorem (cf.~\cite{iph}) we obtain the following extracted terms:
\begin{eqnarray*}
 \etdp\PP &:=& \RecNat\,n\,\etdp{M}\,(\lam {n,x_v} \etdp{N})\\
 \etdn\PP &:=& \RecNat\,n\,\etdn{M} \left(\lam {n,p}\;
\contract{u}{(\etdn{N}\xi)}{p}\right),\quad\text{ for }\xi := \subst{x_v}{\etdp\PP},
\end{eqnarray*}
where $\contract{}{}{}$ is a case distinction operator, defined as follows:
\begin{align*}
\contract{u}{t_1}{t_2} := 
\begin{cases}
 t_1,&\text{ if }u:C\notin \FA(N),\\
 t_2,&\text{ if }u:C\notin \FA(M),\\
\ite{\intd{C}{x_u}{t_1}}{t_2}{t_1},&\text{ otherwise.}
\end{cases}
\end{align*}
Note that the case distinction operator depends not only on the assumption $u$, but on the proof branches $M$ and $N$ as well. To keep notation simpler we do not make this dependency explicit; the referred proof will be clear from the context.

We first notice that the computation of $\etdp\PP$ is linear on $n$. However, in each recursive step of $\etdn\PP$ in order to compute $\etdn N$, we invoke a sub-computation of $\etdp\PP$ for the current value of $n$. This makes the computation of $\etdn\PP$ at least quadratic on $n$. In the general case of treatment of induction this cannot be avoided. However, in the case where $\typen A = \ve$ it is easy to see that we can compute positive and negative content simultaneously:
\begin{align*}
 \pair{\etdp\PP,\etdn\PP} &:= \RecNat\,n\,\pair{\etdp M, \etdn M}\\
 &\qquad\qquad \left(\lam {n,x_v,\pair{p_+,p_-}} \pair{\etdp{N}p_+, \contract{u}{(\etdn{N}\xi')}{p_-}}\right),
\end{align*}
where $\xi' := \subst{x_v}{p_+}$. Thus, as shown in \cite{quasilin}, by avoiding recomputation we might improve worst time complexity of the program. However, in this special case we can optimise even further. For a fixed $n$, $\etdn\PP$ can be seen as performing a linear search for a counterexample for $C$ among the $n$ candidates in the list $L^n := (\etdn{M}, (\etdn{N}\xi'\subst nk)_{k<n-1})$. Formally,
\begin{gather*}
\intd Cx{\etdn\PP} \leqv \bigwedge_{k<n}\intd Cx{L^n_k}\quad \text{ and }\quad \ex {K<n}\,\etdn\PP = L^n_K.
\end{gather*}

The definition of $\contract{}{}{}$ is asymmetric: it performs the counterexample check on one of its operands only (cf. \cite{schwichtwfi}). In the considered case, $\etdn\PP$ always returns the \emph{last possible counterexample} in the list $L^n$, i.e., $\all {k>K} \intd Cx{L^n_k}$. This behaviour seems inefficient and a simple idea to change it is to reverse the operands of $\contract{}{}{}$ in the definition of $\etdn\PP$. Indeed, in this case we would return the \emph{first possible counterexample} from $L^n$, i.e., $\all {k<K} \intd Cx{L^n_k}$. Unfortunately, this will not improve the efficiency of the extracted program, because it will still perform $n$ steps, computing all elements from the list $L^n$ and performing $n-1$ case distinctions. It is clearly sufficient to terminate the recursion as soon as we find the first index $K$ for which $\neg \intd Cx{L^n_K}$. Although this will not change the \emph{worst time complexity} of the program, it might improve the \emph{average time complexity}, in case the expected value of $K$ is lower than $O(n)$.

Such an earlier terminating search could be implemented by adding a boolean flag $b$, which specifies whether a counterexample is already found. For example:
\newcommand{\checkcontract}[4]{#3\,{\stackrel{#1,#2}{\ltimes}}\,#4}
\begin{align*}
 \pair{\etdp\PP,\pair{\etdn\PP,b}} &:= \RecNat\,n\,\pair{\etdp M, \pair{\etdn M,\false}}\\
 & \left(\lam {n,x_v,\pair{p_+,\pair{p_-,b}}} \pair{\etdp{N}p_+, \;\checkcontract{b}{u}{p_-}{\etdn{N}\xi'}}\right),\\
 \text{where }\checkcontract{b}{u}{t_1}{t_2}&:= \ite b {\pair{t_1,\true}} {\big(\ite {\intd Cx{t_1}}{\pair{t_2,\false}}{\pair{t_1,\true}}\big)}.
\end{align*}

Note that the assumption $\typen A = \ve$ is important, otherwise $p_-$ would be a function, applied to a term depending on $n$ on each recursive step. This would prevent us from using the information that a counterexample is found on an earlier step to terminate the recursion.

In \cite{iph} Ratiu and the author considered the Infinite Pigeonhole Principle as a case study for program extraction from non-constructive proofs. There we showed that the refined $A$-translation method \cite{bbs} extracts a program, which has exponential worst time complexity, but polynomial average time complexity, while the program extracted by Dialectica is exponential in both the worst and the average case.  However, the optimisation described in this section applies and we can obtain a Dialectica program, which has polynomial average time complexity, like the one extracted by refined $A$-translation.

The considered early termination of the recursive process is very reminiscent of an abortive control operator \cite{felleisen}, where immediate transfer of the program flow control occurs. As discussed in \cite{iph}, similar situations occur with programs, extracted by refined $A$-translation. For the case study discussed there, it seemed that this feature had an important contribution in achieving better average time complexity. This suggestion is reaffirmed by the fact that adding such an optimisation to the Dialectica interpretation has the same favourable effect on extracted programs.

\section{Marked counterexamples}
\begin{markcounter}
As was discussed in Section \ref{subsec:countercheck}, the programs extracted with the original Dialectica interpretation do not take advantage of the information about the validity of counterexample where a case distinction is needed. The case distinction construction $\contract{u}{t_1}{t_2}$ forces us to choose between two candidate counterexamples $t_1$ and $t_2$ for the assumption $u:C$. The choice is made by direct checking of the decidable Dialectica translation of the formula $C$ for one of the counterexamples. What is not taken into account is that if the check confirms the existence of a counterexample, all further computation of witnesses and counterexamples is pointless. In a certain sense, this can be viewed as avoiding both
\begin{enumerate}
 \item \emph{recomputation} --- the validity of the counterexample is rechecked if we have more than two occurrences of the assumption $C$,
 \item \emph{redundant computation} --- all further counterexamples and witnesses computed are not needed for a sound verification proof.
\end{enumerate}

It is important to note that the common context approach from \cite{quasilin} seems inapplicable for avoiding such kind of recomputation. The reason is the underlying difference between repeated subterms and the recomputation considered here. We can detect duplicated terms during the extraction process and we use a shared context to avoid it. However, the counterexample decision happens during the evaluation of the program and, depending on the input parameters, recomputation might or might not occur. Attempting to use a shared context would imply precomputation of all possible case distinctions, which could be much worse than recomputing only one case distinction.

We will thus follow a different idea. As was already hinted in Section \ref{subsec:countercheck}, an additional marker will be attached to each extracted counterexample, carrying information about its validity. We will use $\marktype$ with three constants: $\true, \false$ and $\bot$, as a type for markers. 
\begin{definition}
 For a formula $A$ we will re-define the positive and negative computational types denoting the new variants as $\typep A$ and $\typen A$. We will also denote $\typea A := \typen A \To \typep A$ and $\typem A := \marktype \Times \typen A$.
We define:
\begin{align*}
 &\typep {\atom b} := \ve,&&\typen{\atom b} := \ve,\\
 &\typep{A\imp B} := \typep B \Times \typem A,&&\typen{A\imp B} := \typea A \Times \typen B\\
 &\typep{\all {x^\sigma} A} := \typep A,&&\typen{\all {x^\sigma} B} := \sigma \Times \typen B
\end{align*}
\end{definition}
For clarity $\markby t m$ will denote that $t$ is marked by $m$. Consequently, when we write $\markby t m \redeq s$, we will mean that $m \redeq s\pl$ and $t\redeq s\pr$. The marker constants have the following intended meaning:
\begin{itemize}
 \item $\markby t \bot$ --- we have no information yet about the validity of $\intd {C_i}{x_i}t$,
 \item $\markby t \false$ --- we have checked that $\neg \intd {C_i}{x_i}t$,
 \item $\markby t \true$ --- $t$ is an arbitrarily chosen term and we should prefer another candidate counterexample without the need to check $\intd {C_i}{x_i}t$.
\end{itemize}
The change in the positive type in the implication case of the translation leads to a slight adjustment to the Dialectica translation (emphasized by a box below):
\begin{align*}
 \intd{A\imp B}fz := \intd A{z\pl}{fz\pr\;\text{\fbox{$\pr$}}} \imp \intd B{(\extapply f{z\pl})\pl}{z\pr}.
\end{align*}
The essential use of the marker comes in the definition of case distinction terms.
\begin{lemma}
\label{lem:markcontract}
  For every formula $C$ and variable $x:\typea C$ there is a term $\contractterm C:\typem C \To\typem C \To \typem C$ with $\FV(\contractterm C)\subseteq \FV(C) \cup \set x$, such that for $t_1,t_2 : \typem C$ from the assumptions $u_i:(m_i = \false \imp \neg\intd Cx{s_i})$ we can prove
 \begin{enumerate}
  \item[$A_i:$] $(m\neq \true \imp \intd Cxs) \imp (m_i \neq\true \imp \intd Cx{s_i})$,
  \item[$B:$] $m = \false \imp \neg\intd Cxs$,
 \end{enumerate}
where $\markby{s_i}{m_i} \redeq t_i$ and $\markby sm \redeq \contractterm Ct_1t_2$.
\end{lemma}
\begin{proof}
  Define
\begin{align*}
  \contractterm C(\markby{s_1}{m_1})(\markby{s_2}{m_2}) &:= 
  \ite{\big(T_\vee(m_2 = \true)(m_1 = \false)\big)}{t_1}\\
  &\quad\;\;\ite{\big(T_\vee(m_1 = \true)(T_\vee(m_2 = \false)(T_Cxs_1))\big)}{t_2}{(\markby{s_1}\false)},
\end{align*}
where $T_\vee := \lam {x,y} \ite x\true y$ and $T_C$ is such that $\atom{T_C x y} \leqv \intd C x y$. It is clear that
\begin{align*}
 \atom{T_\vee x y} \leqv (\neg \atom x \imp \atom y) \leqv (\neg \atom y \imp \atom x).
\end{align*}

Let us denote $D_i := (\contractterm C{t_1}{t_2} = t_i)$. It is easy to see that
\begin{itemize}
 \item $D_i$ immediately implies $A_i$ and also $B$ by $u_i$,
 \item $m_i = \true$ immediately implies $A_i$
 \item $D_i \land (m_i = \false)$ implies $\neg \intd Cxs$ by $u_i$, which contradicts with the premise $m \neq \true \imp \intd Cxs$, implying both $A_1$ and $A_2$.
\end{itemize}
Table \ref{table:cases} summarizes the validity of $D_i$ depending on the values of $m_1$ and $m_2$. It can be checked that the arguments above are sufficient to establish the validity of $A_1, A_2$ and $B$ in all cases except the one marked by '?'. In order to complete the proof we assume that $m_1 = m_2 = \bot$ and consider cases on the decidable formula $\intd Cx{s_1}$.
\begin{table}
\begin{gather*}
 \begin{array}{cc|c|c|c|}
  \cline{3-5}
  &&\multicolumn{3}{|c|}{m_1} \\
  \cline{3-5}
   && \true & \false & \bot\\
  \hline
  \multicolumn{1}{|c|}{\text{\multirow{3}{*}{$m_2$}}} & \true & D_1 & D_2 & D_2\\
  \cline{2-5}
  \multicolumn{1}{|c|}{}&\false & D_1 & D_1 & D_1\\
  \cline{2-5}
  \multicolumn{1}{|c|}{}&\bot & D_1 & D_2 & ?\\
  \hline
 \end{array}
\end{gather*}
\caption{Case analysis for $\contractterm C{t_1}{t_2}$}\label{table:cases}
\end{table}

\emph{Case $\intd Cx{s_1}$}. $D_2$ holds, thus we have only $A_1$ to prove. However, the conclusion of $A_1$ is exactly what we assumed in this case.

\emph{Case $\neg \intd Cx{s_1}$}. We have $s = s_1$, which implies that $\neg \intd Cxs$, proving $B$. On the other hand, $m = \false$, which contradicts with the premise $m \neq \true \imp \intd Cxs$, implying both $A_1$ and $A_2$.

Finally, assuming that $T_C$ is a variable bound by an external definition context as in \cite{quasilin}, we see that $\size{\contractterm{C}{t_1}{t_2}} \leq \size{t_1} + \size{t_2} + K$ for some constant $K$, not dependent on $C$, $t_1$ or $t_2$.
\end{proof}

We will prove soundness for the modified variant of the Dialectica interpretation using the constructions $\et M, \etdp M, \etdn M$ from Theorem \ref{thm:quasilin}. Their types will be as follows:
\begin{gather*}
 \et M:\typen A \imp \holetype,\qquad \etdp M:\typep A, \qquad \etdn M_i:\typem {C_i}.
\end{gather*}

\begin{theorem}[Soundness of counterexample marking]\label{thm:countmark}
   Let $\PP:A$ be a proof from assumptions $u_i:C_i$. Let $x_i:\typea{C_i}$ and $y_A:\typen A$ be fresh variables. Then there is a term $\etds \PP$, satisfying the following conditions:
\begin{enumerate}
\enumr
 \item we can prove $\intd A {\etdsp \PP}{y_A}$ from $(m_i \neq \true \imp \intd {C_i} {x_i} {s_i})$,
 \item we can prove $m_i = \false \imp \neg \intd {C_i} {x_i} {s_i}$,
 \item $\FV(\etds \PP) \subseteq \FV(\PP) \cup \set{x_i}$,
 \item $\size{\etds \PP} = K(\size \PP + {\msl \PP}^2)$ for a fixed constant $K$, not depending on $\PP$,
\end{enumerate}
where $\markby{s_i}{m_i} \redeq \etdsn\PP_i y_A$.
\end{theorem}
\begin{proof}
 The proof is a modification of the argument needed for Theorem \ref{thm:quasilin}.
 
 \emph{Case $u:A$}. We set as before $\et \PP := \lam {y_A} \hole$, $\etdp \PP  := x_u y_A$ and set $\etdn \PP_u := \markby{y_A}\bot$. Then $\etdsp \PP \redeq \lam {y_A} x_u y_A \redeq x_u$ and $\etdsn \PP_i {y_A} \redeq \markby{y_A}\bot$. The assumption premise $m_u \neq \true$ holds, which is enough to conclude that $\intd A{x_u}{y_A}$. On the other hand $m_u \neq \false$, which makes the second condition trivially true. The size bounds and the variable condition also hold as in Theorem \ref{thm:quasilin}.
 
 \emph{Case $\lam {u^B} M^C$}. Let us denote $\markby{s_i}{m_i} \redeq \etdsn M_iy_C$ for $i$ ranging over all assumption variables of $M$, including $u$. By induction hypothesis we have a proof of $\intd C {\etdsp M}{y_C}$ from $m_u \neq \true \imp \intd B{x_u}{s_u}$ and $m_i \neq \true \imp\intd {C_i} {x_i} {s_i}$, as well as proofs of $m_u = \false \imp \neg \intd B{x_u}{s_u}$ and $m_i = \false \imp \neg \intd {C_i}{x_i}{s_i}$. The extracted terms which work for Theorem \ref{thm:quasilin} are still applicable:
 \begin{alignat*}2
  &\et\PP &&:= \lam {y_A}\lett{x_u := y_A\pl}{\et M(y_A\pr)},\\
  &\etdp \PP &&:= \pair{\etdp M, \etdn M_u},\\
  &\etdn\PP_i &&:= \etdn M_i.
 \end{alignat*}
 However, we consider an additional special case: if $u\notin \FA(M)$, then we set $\etdn M_u := \markby{\inhab}{\true}$, where $\inhab$ is an arbitrary term of type $\typen B$.
 Substituting $y_A$ with $\pair{x_u,y_C}$, it suffices to prove
 \begin{enumerate}
  \item\label{enum:first} $\intd B{x_u}{s_u} \imp \intd C{\etdsp M}{y_C}\text{ from }m_i \neq \true \imp \intd {C_i} {x_i} {s_i}$,
  \item\label{enum:second} $m_i = \false \imp \neg \intd {C_i}{x_i}{s_i}$.
 \end{enumerate}
\ref{enum:second} follows directly from our induction hypothesis and for \ref{enum:first} we consider subcases on $m_u$. The case $m_u = \bot$ is proved as in \cite{quasilin}. If $m_u = \false$ then by induction hypothesis we have $\neg \intd B{x_u}{s_u}$ and we can conclude using $\efq$. Finally, if $u\notin \FA(m)$ or $m_u = \true$, then by induction hypothesis we can actually prove $\intd C{\etdsp M}{y_C}$ without using the assumption $m_u \neq \true \imp \intd B{x_u}{s_u}$. Therefore we can conclude by using the same proof with a void implication introduction.

\emph{Case $M^{B\imp A} N^B$}. We define $\et\PP, \etdp \PP$ and $\etdn \PP_i$ essentially as in Theorem \ref{thm:quasilin}, with the only difference that we redefine the case distinction operator $\contract {}{}{}$ used to combine the negative extracted terms of type $\typem {C_i}$ as follows:
\begin{alignat*}2
  &\contract {u_i}{t_1}{t_2} &&:= \begin{cases}
                             t_1,&\text{ if }u_i\notin\FA(N),\\
			      t_2,&\text{ if }u_i\notin\FA(M),\\
                             \contractterm {C_i}t_1t_2,&\text{ otherwise.}
                            \end{cases}
\end{alignat*}
The results from Lemma \ref{lem:markcontract} are sufficient to conclude the proof.

\emph{Cases $\lam {x^\rho} M^B$ and $M^{\all {x^\rho} A} t^\rho$}. The proof of the same case in Theorem \ref{thm:quasilin} still applies, because in both cases we neither remove nor introduce assumptions.

\emph{Case $\Ind_{\bool,A(b)}b\,\,M^{A(\true)}\,N^{A(\false)}$}. We define the extracted terms as in Theorem \ref{thm:quasilin} with the only change that $\contract {}{}{}$ is again defined as above. Lemma \ref{lem:markcontract} allows us to apply the usual soundness proof for this case.

\emph{Case $\Ind_{\nat,A(n)}\,n\,M^{A(0)}\,(\lam {n,u^{A(n)}} N^{A(n+1)})$}. By induction hypothesis we have:
\begin{alignat*}2
 \text{ a proof }M'\text{ of }\quad &\intd {A(0)}{\etdsp M}{y_A}\qquad&&\text{from}\quad m_i \neq \true \imp \intd {C_i}{x_i}{s_i},\\
 \text{ proofs }M''_i\text{ of }\quad &m_i = \false \imp \neg \intd {C_i}{x_i}{s_i},\\
 \text{ a proof }N'\text{ of }\quad&\intd {A(n+1)}{\etdsp N}{y_A}\qquad&&\text{from}\quad n_i \neq \true \imp \intd {C_i}{x_i}{r_i}\text{ and }\\
 &&&\phantom{\text{from}}\quad n_u \neq \true \imp \intd {A(n)}{x_u}{r_u},\\
 \text{ proofs }N''_i\text{ of }\quad& n_i = \false \imp \neg \intd {C_i}{x_i}{r_i},\\
 \text{ a proof }N''_u\text{ of }\quad& n_u = \false \imp \neg \intd {A(n)}{x_u}{r_u},
\end{alignat*}
where $\markby{s_i}{m_i} \redeq \etdsn M_i y_A$ and $\markby {r_j}{n_j} \redeq \etdsn N_j y_A$ for $i$ and $j$ ranging over $\FA(M)$ and $\FA(N)$, respectively.

We define the extracted terms $\et\PP, \etdp \PP$ and $\etdn \PP_i$ as in Theorem \ref{thm:quasilin}, but using the marker-aware variant of $\contract{}{}{}$. We have
\begin{alignat*}2
\tag{$*$}\label{eq:defl}
 &\etds \PP\subst n0 &&\redeq \etds M\\
 &\etdsp \PP\subst{n}{n+1} &&\redeq \lett{x_u := \etds \PP}\etdsp N,\\
 &\etdsn \PP_i \subst{n}{n+1}y_A &&\redeq \lett{x_u := \etds\PP}\contract {u_i}{\etdsn N_iy_A}{\etdsn\PP_i(\etdsn N_uy_A)}.
\end{alignat*}

Let us denote $\markby{t_i}{p_i} \redeq \etdsn \PP_i y_A$. To prove soundness of the term $\etds \PP$ we will use induction on $n$ to prove the formulas
\begin{align*}
 F(n) &:= \all {y_A} \bigg(\Big(\bigwedge_i \big(p_i \neq \true \imp \intd {C_i}{x_i}{t_i}\big)\Big) \imp \intd A {\etdsp \PP} {y_A}\bigg),\\
 G_i(n) &:= \all {y_A} \Big( p_i = \false \imp \neg \intd {C_i}{x_i}{t_i} \Big).
\end{align*}

For $n=0$ we can directly use the proofs $M'$ and $M''_i$ from the induction hypothesis, with the necessary implication introductions.
Now let us assume $F(n)$ and $G_i(n)$, fix $y_A$ and the premises $p_i^+ \neq \true \imp \intd {C_i}{x_i}{t_i^+}$, where $\markby {t_i^+}{p_i^+} \redeq \etdsn \PP_i \subst n{n+1}y_A$. By \eqref{eq:defl} and the properties of $\contract{}{}{}$ from Lemma \ref{lem:markcontract} we can conclude
\begin{align*}
 &\Big(n_i\Xi \neq \true\Big) \imp \intd {C_i}{x_i}{r_i\Xi},\tag{1}\label{eq:c1}\\
 &\Big(p_i\subst{y_A}{r_u\Xi} \neq \true\Big) \imp \intd {C_i}{x_i}{t_i\subst{y_A}{r_u\Xi}},\tag{2}\label{eq:c2}
\end{align*}
where $\Xi := \subst{x_u}{\etdsp\PP}$.

$G_i(n+1)$ can be shown from the proofs $N''_i$ and $G_i(n)$ instantiated with $y_A := r_u\Xi$ by using \eqref{eq:defl} and Lemma \ref{lem:markcontract}. To prove $F(n+1)$ we start by using the induction hypothesis $F(n)$ for $y_A := r_u\Xi$ and \eqref{eq:c2} to obtain $\intd A{\etdsp \PP}{r_u\Xi}$. In order to continue further, we need to consider subcases on $n_u\Xi$.

In case $n_u\Xi = \bot$, by \eqref{eq:c1} we have all premises of $N'\Xi$, thus we can conclude $\intd A{\etdsp \PP\subst{n}{n+1}}{y_A}$, which was to be shown.

In case $n_u\Xi = \false$, we can use $N''_u\Xi$ to derive a contradiction and conclude by using $\efq$.

In case $n_u\Xi = \true$, the premise $n_u\Xi \neq \true \imp \intd {A(n)}{x_u}{r_u\Xi}$ of $N'\Xi$ is trivially true and by \eqref{eq:c1} we have all other premises, hence 
$\intd A{\etdsp \PP\subst{n}{n+1}}{y_A}$.
\end{proof}

\end{markcounter}
\end{linextr}

\section{Conclusion and future work}
The presented variant of the Dialectica interpretation interleaves extracted programs with additional information, which is utilised during evaluation in order to omit redundant calculations. Other kinds of redundancies can be avoided by means of uniform annotations, as described in \cite{ldrev,decorating}. It can be argued that extensions of such technical nature may obscure the obtained computational content. However, such ideas seem to be practically applicable when implementing an automatic extraction method, with the goal to find a correct program, which is not less efficient than a non-verified hand-written program. A topic of further investigation would be to find with a suitable combination between the refined uniform annotations described in \cite{dialfine} and the current extension of the interpretation.


\end{document}